\documentclass[aps,prl,floatfix,twocolumn,amsfonts,amsmath,amssymb,nofootinbib,reprint,superscriptaddress,longbibliography,noeprint]{revtex4-2}
\usepackage{graphicx} 
\usepackage{xcolor}
\pdfoutput=1
\usepackage{calligra}
 
\usepackage{graphicx}
\usepackage{dcolumn}
\usepackage{bm}
\usepackage[utf8]{inputenc}
\usepackage{booktabs}
\usepackage[colorlinks]{hyperref}
\usepackage[normalem]{ulem}
\usepackage{comment}
\usepackage{xcolor}
\usepackage{amsthm}
\usepackage{mathtools}
\usepackage{bbm}
\usepackage{bm}
\usepackage{graphicx}
\usepackage{enumerate}
\usepackage{tikz}
\usetikzlibrary{arrows}
\usepackage{soul}
\usepackage{csquotes}

\newtheorem{theorem}{Theorem}
\newtheorem{proposition}{Proposition}

\newtheorem{corollary}[]{Corollary}

\newtheorem{definition}{Definition}

\def\ketbra#1#2{\mathinner{|{#1}\rangle\!\langle{#2}|}}

\def\a{\bm{a}}
\def\x{\bm{x}}
\def\w{\bm{w}}
\DeclareMathOperator{\Tr}{Tr}

\renewcommand\L{\mathcal{L}}

\newcommand{\HS}{\mathcal{H}}
\newcommand{\N}{\mathcal{N}}
\newcommand{\A}{\mathcal{A}}

\newcommand{\K}{\mathcal{K}}

\newcommand{\X}{\mathcal{X}}

\newcommand{\I}{\mathcal{I}}

\newcommand{\id}{\mathbbm{1}}

\renewcommand{\a}{\bm{a}}

\begin{document}

\preprint{APS/123-QED}
\title{What makes a causal loop consistent?}
\author{Hippolyte Dourdent}
 \email{hippolyte.dourdent@uni-muenster.de}
\affiliation{ICFO-Institut de Ciencies Fotoniques, The Barcelona Institute of Science and Technology,\\ 08860 Castelldefels, Barcelona, Spain}\affiliation{Department for Quantum Technology, Universität Münster,
Heisenbergstraße 11, 48149 Münster, Germany}

\author{{Andreas Leitherer}}
\affiliation{ICFO-Institut de Ciencies Fotoniques, The Barcelona Institute of Science and Technology,\\ 08860 Castelldefels, Barcelona, Spain}

\author{{Emanuel-Cristian Boghiu}}
\affiliation{Barcelona Supercomputing Center, Plaça d'Eusebi Güell, 1-3 08034, Barcelona, Spain}
\affiliation{Fakult\"at f\"ur Mathematik, Universit\"at Wien, Oskar-Morgenstern-Platz 1, 1090 Vienna, Austria}

\author{{Kyrylo Simonov}}
\affiliation{Fakult\"at f\"ur Mathematik, Universit\"at Wien, Oskar-Morgenstern-Platz 1, 1090 Vienna, Austria}

\author{{Ravi Kunjwal}}
\affiliation{Aix-Marseille University, CNRS, LIS, Marseille, France}

\author{Antonio Acín}
\affiliation{ICFO-Institut de Ciencies Fotoniques, The Barcelona Institute of Science and Technology,\\ 08860 Castelldefels, Barcelona, Spain}
\affiliation{ICREA, Passeig Lluis Companys 23, 08010 Barcelona, Spain}

\date{\today}
\begin{abstract}
Causal loops escape paradoxes if they generate valid probabilities under arbitrary independent local interventions. We show that a deterministic classical process is logically consistent if and only if it contains no signaling loops: cyclic signaling is forbidden but cyclic causation is not, allowing indefinite causal order. Logical consistency is then found to hold exactly when the list of events of a process is pairwise exclusive and complete, providing the first event-based, intervention-free, non-recursive characterization of consistent causal loops. Our results repair a previously introduced intervention-free criterion that misses global loops and reveal that what makes deterministic classical communications without causal order consistent is a multipartite form of Specker's principle: if any two of several questions can be answered jointly, so can all of them.
\end{abstract}

\maketitle

\emph{Introduction.}
 The assumption that physical events unfold within a well-defined causal order underlies both classical and quantum physics. Yet, general relativity admits spacetime geometries containing closed timelike curves (CTCs)~\cite{lanczos1924stationare, godel49}. The latter famously give rise to the \textit{grandfather paradox}, a (hypothetically) physical contradiction induced by the suppression of a cause by its own effect. What principle, then, prevents causal loops from generating such contradictions? Hawking's \textit{chronology protection} conjecture~\cite{hawking} answers this question by forbidding CTCs altogether, whereas Novikov's \textit{self-consistency} principle permits them, provided that events along the loop are self-consistent, i.e.\ never lead to a contradiction~\cite{novikov, novikov1989analysis}. 

 Self-consistency, however, introduces a second puzzle. A causal loop may contain an effect that is also its own cause, producing an apparently self-originating piece of information, the \textit{bootstrap paradox}. Although such loops might be dismissed as a benign underdetermination, information-theoretic considerations show that, under the principle of \textit{freedom of choice}\textemdash the ability of local agents to choose their interventions independently\textemdash bootstrap and grandfather paradoxes become operationally equivalent~\cite{baumeler21}. A satisfactory notion of consistency must therefore exclude not only contradictions, but also tautologies.

This criterion is captured by \textit{logical consistency}, the information-theoretic requirement that a communication process should always generate valid joint probabilities under arbitrary, independent local operations. Surprisingly, logical consistency does not require a fixed global causal order. It permits indefinite causal structures, including correlations that violate causal inequalities~\cite{costa26}. Such \textit{non-causal} behaviors do not solely arise from quantum \cite{oreshkov1} or  probabilistic \cite{kunjwal23} dynamics. They also manifest in fully classical, deterministic single-round communication channels known as \textit{process functions} \cite{baumeler14,baumeler2}.  Process functions without global past \cite{baumeler14,baumeler2,af,baumeler19,tobar,baumeler22,mejdoub25,dourdent25}, such as the canonical tripartite Lugano process \cite{af,baumeler2}, dynamically route information through a mesh of mutual controls of causal structures, such that no party's input remains independent of the others' outputs, while still preserving freedom of choice.

By definition, logical consistency is the mechanism that shields these causal loops from paradoxes. But can it be characterized by the communication events alone, without invoking the agents or checking every group of them?
Its relation to paradoxes is explicitly formalized through \textit{fixed-point uniqueness} (Theorem \ref{th:fixedpoint}) \cite{baumeler16,baumeler19}: logically consistent causal loops must yield a unique fixed point for any local interventions, while grandfather and bootstrap antinomies corresponding precisely to dynamics with zero or multiple fixed points, respectively. However, evaluating fixed points requires checking the induced dynamics for every choice of local interventions. Instead,
\textit{unipartite reducibility} (Theorem \ref{lem:characpf})~\cite{baumeler19} captures the intuitive idea that reducibility under the free local operations of a single agent provides another core representation of logical consistency. It remains however recursive and intervention-dependent. Inspired by the idea that non-causal process functions rely on a non-trivial mesh of classical controls of causal orders, earlier works sought an intervention-free alternative in \textit{output-reducibility} (Proposition \ref{prop:outputred})~\cite{baumeler19,tobar}\textemdash a criterion that, as we show below, turns out to be incomplete.

\medskip

\emph{Results.} First, we reformulate fixed point uniqueness as an explicit \textit{prohibition of any signaling loop}, whether it is \textit{local, internal, or global} (Corollary \ref{th:noloop}). It reveals that output-reducibility is incomplete, as it fails to rule out global loops. We then show how output-reducibility can be rectified by \textit{non-erasing reducibility} (Theorem \ref{th:nonerasingcharac}), based on reductions that preserve dependencies, detecting all signaling loops. Alternatively, it can be completed with a linear \textit{process-matrix} constraint (Theorem \ref{th:outputrednogloop}), also formulated as function-level \textit{phase-erasure} conditions (Theorem~\ref{th:outputredphase}). In fact, as process functions are classical deterministic process matrices, all linear constraints characterizing the latter can be transformed into a \textit{phase erasure characterization} (Theorem~\ref{th:phaseerasure}). This generalizes to arbitrary local dimension the constraints of the Boolean ``parity-erasure'' principle, recently proposed in Ref.~\cite{liu25} as a fundamental principle underlying indefinite causal order.
Finally, we introduce the first purely event-based characterization of process functions (Theorem~\ref{th:characPE}): a process function is valid if and only if it forms a \textit{complete list of pairwise exclusive events}, where two global events are defined as exclusive whenever some party assigns them distinct outcomes for the same local setting. Thus, consistency of a causal loop only requires the communication structure to be well-defined for every tuple of local outcomes. This provides a complete, conceptually parsimonious, non-recursive and intervention-free characterization of logical consistency. It identifies the latter as a multipartite form of Specker's principle \cite{specker,cabello12s}, the motivational basis \cite{kunjwal2019b,chiribella2014,chiribella2016} of the local orthogonality \cite{fritz13} and exclusivity \cite{cabello3} principles, which play a central role in explaining why quantum correlations are not more nonlocal than they are.
\medskip

\emph{Process functions.} Consider $n$ parties, labeled $k\in\{1,\dots,n\}=\mathcal N$. Each party receives classical inputs $x_k\in\mathcal X_k$ and $i_k\in\mathcal I_k$ and produces classical outputs $a_k\in\mathcal A_k$ and $o_k\in\mathcal O_k$. For compactness, we will use the notations $\mathcal{X}:=\bigtimes_{k=1}^n \mathcal{X}_k$, $\mathcal{A}:=\bigtimes_{k=1}^n \mathcal{A}_k$,   $\bm{x}=(x_1,...,x_n)$, $\bm{a}=(a_1,...,a_n)$, $\bm{a}_{\backslash i}:=(a_1,...,a_{i-1},a_{i+1},..,a_n)$ etc. Here\footnote{Our notation is distinct from that adopted in Refs.~\cite{baumeler2,baumeler19,tobar,kunjwal23a,kunjwal23}.}  $\mathcal X_k,\mathcal A_k$ are exchanged with a shared communication resource\textemdash analogous to hidden variables mediating the correlations in a Bell scenario\textemdash while $\mathcal I_k,\mathcal O_k$ are the local setting and outcome registers. The parties are operationally isolated, interacting only through the shared resource; their settings $i_k$ are freely chosen, and each party applies a local stochastic channel $P(a_k,o_k|x_k,i_k)$ exactly once.
The observed correlations are \begin{equation}
    P(\bm{o}|\bm{i}) = \sum_{\bm{a},\bm{x}} \Bigl(\prod_{k=1}^n P(a_k,o_k| x_k,i_k)\Bigr) P(\bm{x}| \bm{a}), \label{eq:correlation2}
\end{equation} where $P(\bm{x}|\bm{a})$ is the classical process implemented by the resource. The process is \textit{logically consistent} if Eq.~\eqref{eq:correlation2} yields a valid distribution, i.e.\ $P(\bm{o}|\bm{i})\ge 0$, $\sum_o P(\bm{o}|\bm{i})=1$, for every choice of local operations.

\begin{definition}[\emph{\textbf{Process function}}] \label{def:pf}
    A deterministic logically consistent process defines a map $\bm{w}$ called the \textbf{\textit{process function}}, composed of a set of local functions $\bm{w}\coloneq (w_k)_k$ where $w_k:\mathcal{A}\rightarrow\mathcal{X}_k$,
which specifies the relation between outputs and inputs in Eq.~\eqref{eq:correlation2} via $P(\bm{x}|\bm{a})=\delta_{\bm{x},\bm{w}(\bm{a})}$, where $\delta_{i,j}$ is the Kronecker delta. Similarly, deterministic local channels, also called \textbf{\textit{local interventions}}, define  functions $\bm{f}:=(f_k)_k$ where $f_k : \mathcal{X}_k \times \mathcal{I}_k \to \mathcal{A}_k \times \mathcal{O}_k$ such that $P(a_k,o_k|x_k,i_k)=\delta_{(a_k,o_k),f_k(x_k,i_k)}$.
\end{definition}

\medskip

\emph{Fixed point uniqueness.} Let us consider the simplified scenario where $|\mathcal{I}_k|=|\mathcal{O}_k|=1$ for all parties $k$. Logical consistency then reduces to $\sum_{\x} \prod_k \delta_{x_k,w_k(\bm{f}(\x))}=1$, equivalently $\sum_{\a}\prod_k\delta_{a_k,f_k(w_k(\a))}=1$ for any local interventions $\bm{f}$. This yields the following characterization:

\begin{theorem}[\emph{\textbf{Fixed point uniqueness}}]\label{th:fixedpoint} A quasi-process function $\bm{w}:\mathcal{A}\rightarrow\mathcal{X}$, that is, a set of local functions $\bm{w}\coloneq (w_k)_k$ where $w_k:\mathcal{A}\rightarrow\mathcal{X}_k$, is a process function if and only if for every  $\bm{f}=(f_1,\ldots,f_n),$ $f_k:\mathcal{X}_k\to \mathcal{A}_k$, the composed map $\w\circ \bm{f}$ has a unique fixed point 
    \begin{align}
    &\forall \bm{f}:\mathcal{X}\rightarrow \mathcal{A},\hspace{1mm} \exists!\hspace{1mm} \bm{x^{f}}\hbox{ s.t. } \forall k,\ x_k^f=w_k(\bm{f}(\bm{x^f})),\label{eq:fixedpoint}
\end{align}
or equivalently, $\bm{f}\circ \w$ has a unique fixed point. 
\begin{proof}
    Eq.~\eqref{eq:fixedpoint} is proven in Refs.~\cite{baumeler16,baumeler19}. 
\end{proof}
\end{theorem}  

 Fixed-point uniqueness excludes both the grandfather antinomy (a collection of interventions $\bm{f}$ is such that $\bm{f}\circ \w$ has no fixed point, a self-contradiction, e.g.\ $a_k=a_k\oplus1$) and the bootstrap antinomy (a collection of interventions $\bm{f}$ is such that $\bm{f}\circ \w$ has several fixed points, a tautology, e.g.\ $a_k=a_k$), which are anyway equivalent under local operations~\cite{baumeler21}. It thus acts as a ``time-policing'' principle compatible with freedom of choice. It relies, however, on every possible local intervention.

\medskip

\emph{Unipartite reducibility.}
Another complete characterization~\cite[Lemma~3]{baumeler19} replaces the global composition $\w\circ \bm{f}$ of Theorem \ref{th:fixedpoint} by a composition with a single local operation $f_k$, yielding a reduced process function.

\begin{definition}\label{def:redpf}
    \textit{\textbf{(Reduced process function)}} 
Let $\w:\mathcal A\to\mathcal X$ be non-self-signaling, i.e.\ $w_k(a)=w_k(\a_{\backslash k})$ for all $k$. For a local operation $f_k:\mathcal X_k\to\mathcal A_k$, the reduced process function $w^{f_k}:\mathcal A_{\backslash k}\to\mathcal X_{\backslash k}$ is obtained by feeding $x_k=w_k(a_{\backslash k})$ into $f_k$ and closing the loop (Fig.~\ref{fig:output_reduced_pf}):
$\bm{w}^{f_k}(\bm{a}_{\backslash k}) := \bm{w}_{\backslash k}(\bm{a}_{\backslash k}, f_k(w_k(\bm{a}_{\backslash k}))).$ 
\end{definition}

    \begin{figure}[ht!]
    \centering
\includegraphics[width=0.9\columnwidth]{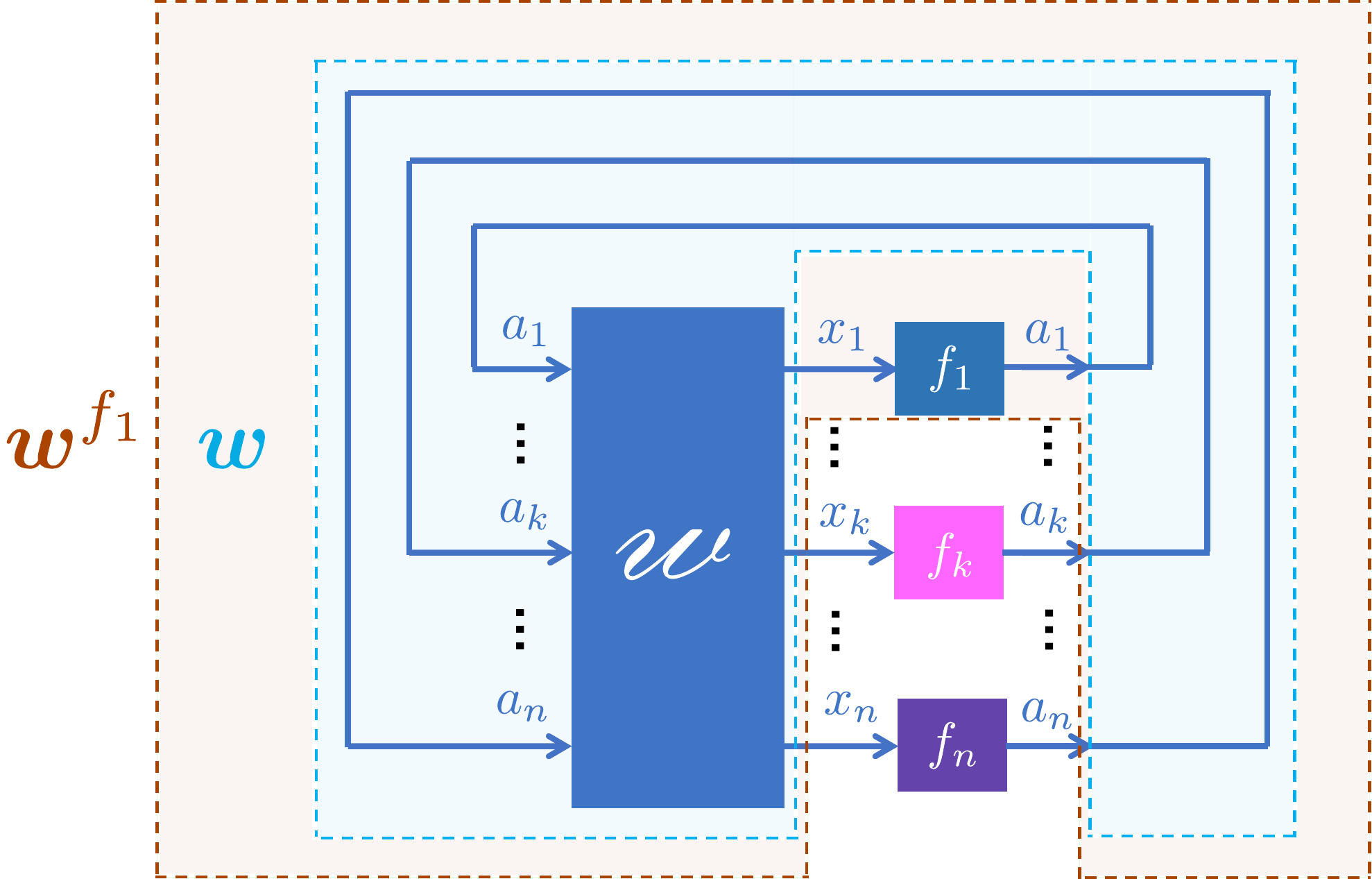}
    \caption{Visualization of the reduced process function as defined in Def.~\eqref{def:redpf} for the case of fixing $f_1$, a local operation of the first party.}    \label{fig:output_reduced_pf}
\end{figure}

\medskip

\begin{theorem}[\emph{\textbf{Unipartite reducibility}}]\label{lem:characpf}   Let $\bm{w}:\mathcal{A}\rightarrow\mathcal{X}$ be a non-self-signaling function. Then $\w$ is an $n-$partite process function if and only if there exists a party $k$ such that $\w^{f_k}$ is a valid $(n-1)-$partite process function for every $f_k$. (Necessity in fact holds for every party $k$.)
\end{theorem}
\begin{proof}
    Proven in Ref.~\cite[Lemma~3]{baumeler19}.
\end{proof}
While losing the explicit relation with the prohibition of paradoxes, unipartite reducibility is structurally informative. It expresses the idea that logical consistency is conserved when a party is absorbed into the process. It is however hard to apply, as it is recursive and quantifies over \emph{all} local operations $f_k$.

\medskip

\emph{Output reducibility.} Refs.~\cite{baumeler19,tobar}
proposed a simplification using only constant interventions
$f_k(x_k)=a_k$ for all $x_k$, where party $k$ feeds to the shared communication resource the same value $a_k$ irrespective of the input $x_k$ received from the resource. The resulting \emph{output-reduced function} is
\begin{equation}
  \w^{a_k}(\a_{\backslash k}) := \big(w_j(\a_{\backslash k},a_k)\big)_{j\ne k}.
  \label{eq:output_reduced}
\end{equation}

\begin{proposition}[\emph{\textbf{Output-reducibility}}]\label{prop:outputred}
    A quasi-process function $\w:\mathcal A\to\mathcal X$ is a process function only if $\w^{a_k}$ is a valid $(n{-}1)-$partite process function for every party $k$ and every $a_k\in\mathcal A_k$.
\end{proposition}

\begin{proof}
    Necessity follows directly from Theorem~\ref{lem:characpf}. 
\end{proof}

Crucially, this condition is operationally intervention-free: one can fix outputs without invoking local operations. Nevertheless, output-reducibility is \emph{not} sufficient: it may incorrectly classify certain logically inconsistent quasi-processes~\cite{kunjwal23} as valid. The simplest counterexamples are the Guess-Your-Neighbor's-Input (GYNI)~\cite{almeida2010guess} processes,
$\w:\mathcal A\to\mathcal X$ where each component $w_k$ depends on a single output $a_{j\ne k}$, and every party's output is fed into exactly one other party's input, forming a signaling loop through all parties. A standard tripartite example is (cf. Fig.~\ref{fig_noloop})
\begin{equation}
  \w_{3\text{-GYNI}}:\quad x_1:=a_3,\ \ x_2:=a_1,\ \ x_3:=a_2,
  \label{eq:gyni}
\end{equation}
with straightforward $n$-partite generalizations. Such loops satisfy output-reducibility: fixing any output, say $a_1$, yields the bipartite reduction
$\w^{a_1}_{3\text{-GYNI}}(a_2,a_3)=(x_2=a_1,\,x_3=a_2)$, which is one-way signaling, hence valid, and the same holds for the other reductions. 
However, GYNI processes are manifestly inconsistent. For instance, if all parties apply an identity channel $a_k:=x_k$, Eq.~\eqref{eq:gyni} gives
$x_1=a_3=x_3=a_2=x_2=a_1=x_1$, i.e.\ the bootstrap paradox arises.
Thus, this inconsistency evades detection here.

\medskip
 
In order to understand more generally what makes output-reducibility incomplete, let us define a \textit{signaling loop} on a nonempty subset $\mathcal K\subseteq\N$ as follows. Fix the outputs $\bar\a:={\bm a}_{\N\backslash\mathcal K}$ of all parties outside $\mathcal K$, and choose local interventions $f_k:\mathcal X_k\to\mathcal A_k$ for $k\in\mathcal K$. These induce the closed-loop map $L_{\mathcal K}^{\bar{\bm a},\bm f}:\mathcal A_{\mathcal K}\longrightarrow\mathcal A_{\mathcal K},$
\begin{align}
\bigl(L_{\mathcal K}^{\bar{\bm a},\bm f}(\bm a_{\mathcal K})\bigr)_k
=
\bigl(f_k\!\left(
w_k(\bm a_{\mathcal K},\bar{\bm a})
\right)\bigr)_k.
\end{align}
  We say that $\mathcal K$ forms a  signaling loop whenever this map fails to have exactly one fixed point. This establishes a scale of loops, see Fig.~\ref{fig_noloop}, which can be local ($|\mathcal{K}|=1$, a party signaling to itself), internal ($\mathcal{K}\subsetneq\N$, manifested in a reduced function), and global ($\mathcal{K}=\N$, as in Eq.~\eqref{eq:gyni}). Hence, the following corollary can be derived:

\begin{corollary}
[\emph{\textbf{No-$\K$-signaling loop}}]     \label{th:noloop}
A quasi-process function, i.e.\ any tuple $\w=(w_k)$, $w_k:\mathcal A\to\mathcal X_k$, is a process function if and only if, no subset of  $\mathcal{K}\subseteq\N$ parties form a signaling loop, for any fixed remaining outputs.
 \begin{proof}
     By definition, a signaling loop yields paradoxes. The claim therefore follows directly from Theorem~\ref{th:fixedpoint}, and conversely.
 \end{proof}
\end{corollary}

Forbidding signaling loops do not forbid all kinds of cyclic dependencies. In the Lugano process~\cite{af,baumeler2}, 
\begin{align}
    x_1 := a_3(a_2\oplus1),\quad
  x_2 := a_1(a_3\oplus1),\quad
  x_3 := a_2(a_1\oplus1),\label{eq:lugano} 
\end{align}
   each bipartite signaling relation is controlled by the third party, such that no loop survives. Yet, this process function admits no global past, i.e.\ no $x_k$ is constant. 
  Non-causality of process functions therefore rests not on signaling loops, but on a mesh of classically controlled causal
orders.

\begin{figure}[ht!]
	\begin{center}
\includegraphics[width=1\columnwidth]{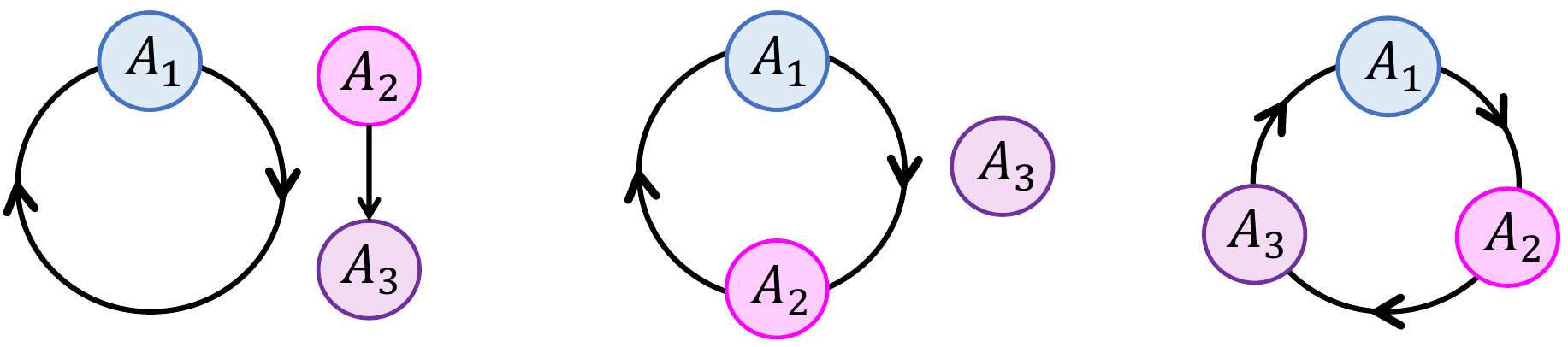}
	\end{center}
	\caption{Examples of three non-valid $n=3$ quasi-processes structures containing (left) a local loop on $A_1$, e.g., $x_1:=a_1$, $x_2:=0 $, $x_3:=a_2$ ; (middle) an internal loop between $A_1$ and $A_2$, e.g., $x_1:=a_2$, $x_2:=a_1$, $x_3:=0$; (right) a global loop, e.g., the GYNI quasi-process $x_1:=a_3$, $x_2:=a_1$, $x_3:=a_2$. For each of these forbidden structures, there exists a set of local interventions that lead to a self-contradiction $a_1=a_1\oplus1$ (e.g. for $f_1:x_1\rightarrow a_1:=x_1\oplus1$, $f_2:x_2\rightarrow a_2:=x_2$ and $f_3:x_3\rightarrow a_3:=x_3$) or a tautology $a_1=a_1$ (e.g. for $f_k:x_k\rightarrow a_k:=x_k$ for all $k$). } 
	\label{fig_noloop}
\end{figure}

Corollary~\ref{th:noloop} sheds light on the incompleteness of output-reducibility. The latter does rule out local and internal loops. A global loop, however, is cut by the very act of fixing an output, and thereby escapes every reduction. This diagnosis comes nevertheless with its cure. 
\medskip

\textit{Non-erasing reducibility.} On the one hand, one can correct output-reducibility at its source, by requiring reducibility under operations that preserve dependencies rather than erase them.

\begin{definition}[\emph{\textbf{Non-erasing operation}}]
\label{def:nonerasing}
A local operation $f_k:\mathcal X_k\to\mathcal A_k$ is
\emph{non-erasing for $\bm w$} if it preserves at least one dependence
on $a_k$: there exist a party $j$, an assignment
$\bm a_{\backslash k}$, and inputs $x_k,x_k'\in\mathcal X_k$ such that 
\begin{align}
     w_j\bigl(f_k(x_k),\bm a_{\backslash k}\bigr)
 \neq
 w_j\bigl(f_k(x_k'),\bm a_{\backslash k}\bigr).
\end{align}
\end{definition}

Two compositions must be distinguished. Testing whether $f_k$ is non-erasing probes the substitution $a_k\mapsto f_k(x_k)$ with
$x_k$ kept free, isolating which dependencies survive.  The
reduction $w^{f_k}$ instead closes the loop, feeding
$x_k=w_k(\a_{\backslash k})$, i.e.\ $a_k\mapsto f_k(w_k(\a_{\backslash k}))$.
This yields:

\begin{theorem}[\emph{\textbf{Non-erasing reducibility}}]\label{th:nonerasingcharac}
A quasi-process function $\w:\mathcal A\to\mathcal X$ is a valid
$n-$partite process function if and only if
\begin{enumerate}
  \item[(i)] for $n=1$, $w$ is constant;
  \item[(ii)] for $n>1$, for every party $k$ and every
  non-erasing operation $f_k:\mathcal X_k\to\mathcal A_k$, the  reduced function $\w^{f_k}:\mathcal A_{\backslash k}\to
  \mathcal X_{\backslash k}$ is well-defined and is a valid $(n{-}1)-$partite process
  function.
\end{enumerate}
\begin{proof}
See Appendix.
\end{proof}
\end{theorem}

\medskip

\textit{Completing output-reducibility.} On the other hand, output-reducibility can be made sufficient by adding a constraint ruling out genuine global signaling loops, i.e.\ non-valid quasi-process that do not contain local nor internal loops. Since process functions are 
equivalent to deterministic classical process matrices, it is natural to consider the latter framework to formalize the missing constraint. By Theorem~4 of Ref.~\cite{baumeler22}, $\w$ can be encoded as a process matrix $W_{\w}$ (Supplemented Material (SM), Eq.~\eqref{eq:classicalPM}). 
In the trace-and-replace notation of Eq.~\eqref{eq:trnrep}, with $A_k^O$ the output space of party $k$, the single linear constraint
\begin{equation}
  {}_{\prod_k[\mathbf 1-A_k^O]}\,W_{\w}=0 ,
  \label{eq:nogloballoop}
\end{equation} forbids
$W_{\w}$ from having component depending jointly on all $n$ outputs. The GYNI process function (Eq.~\eqref{eq:gyni}) violates  Eq.~\eqref{eq:nogloballoop}, while all its output-reductions are valid. Eq.~\eqref{eq:nogloballoop} completes precisely what output-reducibility, which excludes every loop among proper subsets of parties, leaves out: it is violated by any quasi-process whose sole defect is a global loop. It is not, however, a stand-alone diagnostic for global loops. See Appendix.

\begin{theorem}[\emph{\textbf{Output-reducibility and no global loop}}]
\label{th:outputrednogloop}
A quasi-process function $\w:\mathcal A\to\mathcal X$ is a valid $n$-partite
process function if and only if
\begin{enumerate}
  \item[(i)] Proposition \ref{prop:outputred} holds and
  \item[(ii)] ${}_{\prod_k[\mathbf 1-A_k^O]}W_{\w}=0$.
\end{enumerate}
\end{theorem}

\begin{proof}
As $W_{\w}$ is automatically positive and normalized, validity reduces to the linear constraints ${}_{\prod_{k\in\mathcal K}[\mathbf 1-A_k^{O}]\, A_{\N\backslash\mathcal K}^{IO}}W_{\w}=0$ for all
$\emptyset\neq\mathcal K\subseteq\N$~\cite{araujo1,oreshkov16,wechs}. Averaging a reduced process  $\bm w^{a_k}$ over $a_k$ amounts to tracing out the input and output spaces $A_k^{IO}$ of party $k$. Hence
$(i)$ implies the process-matrix constraints for every $\mathcal K\subsetneq\N$, choosing any $k\notin\mathcal K$, while $(ii)$ is the remaining case $\mathcal K=\N$. Conversely, the full family of constraints implies both $(i)$ and $(ii)$. More details can be found in Appendix.
\end{proof}

Because Eq.~\eqref{eq:nogloballoop} is defined on  the  process matrix $W_{\w}$, while $\w$ is purely a classical function, one may ask whether this constraint can be expressed without invoking Hilbert spaces. Indeed, it admits an equivalent, purely classical representation in an explicit \emph{phase-erasure} form:

\begin{theorem}[\emph{\textbf{Output-reducibility and phase erasure}}]\label{th:outputredphase}
A quasi-process function $\w:\mathcal A\to\mathcal X$ is a valid
$n$-partite process function if and only if
\begin{enumerate}
\item[(i)] Proposition \ref{prop:outputred} holds and
  \item[(ii$'$)] for every $\x\in\mathcal X$ and every
  $\bm{\nu}=(\nu_1,\dots,\nu_n)$ with $\nu_k\in\{1,\dots,d_k-1\}$, $d_k:=|\A_k|$,
  \begin{equation}
     \sum_{\bm{a}}
    \delta_{\bm{w}(\bm{a}),\, \bm{x}} \prod_{k=1}^n e^{2\pi i \nu_k\, a_k / d_k} = 0.
    \label{eq:phaseerasure}
  \end{equation}
\end{enumerate}
\begin{proof}
    See SM.
\end{proof}
\end{theorem}

Eq.~\eqref{eq:phaseerasure} is the Fourier transform of Eq.~\eqref{eq:nogloballoop}. For each global setting $\bm x$, the set of outputs producing it carries no Fourier component in which every party has a nonzero frequency. The inputs generated by $\w$ are thus never controlled by all $n$ outputs jointly, and genuine global loops are ruled out.

\medskip

\textit{A complete phase erasure characterization.} Noting that $W_{\w}$ satisfies all process-matrix linear constraints, the same translation applies to each subset $\mathcal K$, yielding a complete family of phase-erasure constraints, which completely characterize process functions:

\begin{theorem}[\emph{\textbf{Phase erasure}}]\label{th:phaseerasure}
A quasi-process function $\w:\mathcal A\to\mathcal X$ is a valid $n$-partite
process function if and only if, for every nonempty
$\mathcal K\subseteq\N$, every $\bm x_{\mathcal K}\in\mathcal X_{\mathcal K}$
and every $(\nu_k)_{k\in\mathcal K}$ with $\nu_k\in\{1,\dots,d_k-1\}$,
\begin{equation}\label{eq:phaserasuretot}
  \sum_{\bm a\in\mathcal A}\ \prod_{k\in\mathcal K}
  \delta_{w_k(\bm a),x_k}\;e^{2\pi i\nu_k a_k/d_k}=0 .
\end{equation}
\begin{proof}
    See SM.
\end{proof}
\end{theorem}

Theorems~(\ref{th:fixedpoint}--\ref{th:phaseerasure}) are multiple faces of the same condition: the deterministic classical special case of the process‑matrix constraints that detects signal loops at every scale. Remarkably, in the Boolean case, $\nu_k=1$ and
Eq.~\eqref{eq:phaserasuretot} reduces to the parity condition
\begin{equation}
  \sum_{\bm a\in\mathcal A}\ \prod_{k\in\mathcal K}\delta_{w_k(\bm a),x_k}\;
  (-1)^{\sum_{k\in\mathcal K}a_k}=0 ,
  \label{eq:parityerasure}
\end{equation}
recovering the constraints on which the \textit{parity-erasure} principle for higher-order processes~\cite{liu25} is based on.

\medskip

All the above criteria ultimately chase signaling loops (and thus paradoxes), in the signaling structure (Corollary~\ref{th:noloop}), through local reductions (Theorems~\ref{th:nonerasingcharac}--\ref{th:outputredphase}) or in the Fourier coefficients (Theorem~\ref{th:phaseerasure}) of $\w$.
All inherit either a recursion over parties or a quantification over their subsets. Is logical consistency expressible without these features?

\medskip

\textit{Pairwise exclusive completeness.}
Here, we find such a criterion, phrased on $\w$ alone, based on a single check over all pairs of events, with no intervention, recursion, nor subset involved.

An \emph{event} is a pair of tuples $(\bm a\,|\,\bm x)\in\mathcal A\times\mathcal X$,
read as ``the parties produce the outputs $\bm a$ upon receiving the inputs
$\bm x$''. Two events are \emph{exclusive} if some party $k$ satisfies
$x_k=x'_k$ and $a_k\neq a'_k$: receiving the same input, it would have to
produce two different outputs, which no local function can do. Exclusive events therefore never occur under a common intervention. A list
$E\subseteq\mathcal A\times\mathcal X$ is \emph{complete} if
$|E|=|\mathcal A|$.

\begin{theorem}[\emph{\textbf{Pairwise-exclusive completeness}}]\label{th:characPE}
A list of events $E\subseteq\mathcal A\times\mathcal X$ is the event list $\{(\bm a\,|\,\w(\bm a))\}_{\bm a\in\mathcal A}$ of an $n$-partite process function if and only if it is complete and pairwise exclusive.
\begin{proof} Two events sharing the same $\bm a$ are never exclusive, so a pairwise exclusive list has pairwise distinct outputs, hence at most $|\mathcal A|$ events. If the list is moreover complete, its outputs exhaust $\mathcal A$, and $E$ is associated with a total function $\w:\mathcal A\to\mathcal X$. Conversely the
event list of a total function is automatically complete. It remains to show that a total quasi-process function $\w$ is a process function if and only if $E$ is pairwise exclusive.

$(\Leftarrow)$ Sufficiency follows from Ref.~\cite[Theorem~3]{dourdent25}. Let $\bm f=(f_k)_k$, $f_k:\mathcal X_k\to\mathcal A_k$, be an arbitrary intervention, and suppose $\bm a\neq\bm a'$ are both fixed points of $\bm f\circ\w$. Exclusivity of the corresponding events gives a party $k$ with $a_k\neq a'_k$ and $w_k(\bm a)=w_k(\bm a')$, hence
\begin{equation}
  a_k=f_k\big(w_k(\bm a)\big)=f_k\big(w_k(\bm a')\big)=a'_k ,
\end{equation}
a contradiction. Every intervention thus admits \emph{at most one} fixed point. Hence the quasi-process function does not generate bootstrap paradoxes. By Ref.~\cite[Theorem~1]{baumeler21}, any total quasi-process function $\w$ that does not suffer from the bootstrap paradox has a unique fixed point. Hence, $\w$ is a process function by Theorem \ref{th:fixedpoint}.

$(\Rightarrow)$  
Let $\bm{w}:\mathcal{A}\rightarrow\mathcal{X}$ be a process function with event list
$E=\{(\bm a|\bm x)\ |\ \bm x=\bm w(\bm a)\}$.
Assume now \textit{ad absurdum} that there exist two events
$(\bm a|\bm x)\neq(\bm a'|\bm x')$ in $E$ violating exclusivity, i.e.\ for every party $k$,
\begin{equation}
  x_k=x'_k \;\Rightarrow\; a_k=a'_k\,
 \text{ and }
  a_k\neq a'_k \;\Rightarrow\; x_k\neq x'_k .
\end{equation}
Let $D=\{k:a_k\neq a'_k\}\neq\varnothing$ because $(\bm{a}|\bm{x})\neq(\bm{a}^\prime|\bm{x}^\prime)$. Define the following intervention for every $k$, 
\begin{equation} \label{eq:def_intervention}
  f_k(x_k)\coloneqq    \begin{cases}a'_k & x_k=x'_k,\\ a_k & \text{otherwise.}\end{cases}
\end{equation}
For $k\in D$ the inputs $x_k\neq x'_k$ are distinct, hence
$f_k(x'_k)=a'_k$ and $f_k(x_k)=a_k$; for $k\notin D$, $f_k$ is the constant map $a_k=a'_k$ and both equalities hold as well. Since
$w_k(\bm a)=x_k$ and $w_k(\bm a')=x'_k$, both $\bm a$ and $\bm a'$ are distinct fixed points of
$\bm f\circ\bm w$ with $\bm f:=(f_k)_k$, contradicting Theorem \ref{th:fixedpoint}. 
 Hence, the process function $\w$ satisfies pairwise exclusivity.
\end{proof}\end{theorem}

\medskip

\begin{table*}[t]
\centering
\setlength{\tabcolsep}{10pt}
\renewcommand{\arraystretch}{1.25}
\begin{tabular}{@{}llcc@{}}
\toprule
\textbf{Characterization} & \textbf{Condition on $\bm w$} &
\textbf{Type} & \textbf{Ref.}\\
\midrule
Fixed-point uniqueness &
$\bm f\circ\bm w$ has a unique fixed point, $\forall\bm f$ &
interventions & \cite{baumeler16}\\
Unipartite reducibility &
$\bm w^{f_k}$ valid for some $k$ and all $f_k$ &
recursive & \cite{baumeler19}\\
\addlinespace[2pt]
\cmidrule(l{0pt}r{0pt}){1-4}
\addlinespace[2pt]
No signaling loop &
no $\mathcal K\subseteq\N$ forms a signaling loop &
signaling structure & Cor.~\ref{th:noloop}\\
Non-erasing reducibility &
$\bm w^{f_k}$ valid $\forall k$ and all non-erasing $f_k$ &
recursive & Thm.~\ref{th:nonerasingcharac}\\
Output-red.\ $+$ no global loop &
$\bm w^{a_k}$ valid $\forall k,a_k$~\cite{baumeler19,tobar}, and
${}_{\prod_k[\mathbf 1-A_k^O]}W_{\bm w}=0$ &
recursive $+$ matrix & Thm.~\ref{th:outputrednogloop}\\
Output-red.\ $+$ phase erasure &
$\bm w^{a_k}$ valid $\forall k,a_k$, and Eq.~\eqref{eq:phaseerasure}
&
recursive $+$ functions & Thm.~\ref{th:outputredphase}\\
Phase erasure &
Eq.~\eqref{eq:phaserasuretot} for all $\emptyset\neq\mathcal K\subseteq\N$ &
functions & Thm.~\ref{th:phaseerasure}\\
Pairwise-exclusive completeness &
$|E|=|\mathcal A|$, all events of $E$ pairwise exclusive &
combinatorial & Thm.~\ref{th:characPE}\\
\bottomrule
\end{tabular}
\caption{Complete characterizations of process functions: those previously
known (top) and those established here (bottom). Only phase erasure and
pairwise-exclusive completeness avoid recursion, and only the latter is a
condition on pairs of events, hence testable in parallel.}
\label{tab:characterizations}
\end{table*}

\textit{Discussion.} The characterizations of process functions brought together in Table~\ref{tab:characterizations} resolve the central question of this work on three distinct levels. 

 On one hand, logical consistency requires forbidding closed signaling loops across all scales (Corollary~\ref{th:noloop}). Notably, this restriction against cyclic signaling does not enforce global causality: cyclic dependence may remain viable if it is mediated by classical control, where each party’s causal placement is dynamically selected by prior outputs (Eq.~\eqref{eq:lugano}).
 
 On the other hand, traditional criteria establish validity based on interventions or a quantification over subsets of parties. In contrast, pairwise-exclusive completeness (Theorem~\ref{th:characPE}) bypasses this requirement: a causal loop is valid precisely when its allowed events form a partition of the output space into mutually exclusive outcomes. By eliminating any direct reliance on loops, interventions, or recursive calls, this property is the simplest to state, inherently non-sequential and thus straightforward to parallelize. Because pairwise exclusivity always holds for complete Boolean event lists, this drastically simplifies exhaustive enumeration\textemdash retrieving the  $12$ bipartite and $744$ tripartite Boolean process functions ($64$ of which possess no global past) corresponding to the deterministic extremal points of the classical processes polytope \cite{baumeler2}.

This criterion recovers the correspondence between process functions and unambiguous complete product bases~\cite{dourdent25}, which specify a list of events that is pairwise exclusive (by unambiguity) and complete. With it comes the relation between process functions and non-signaling inequalities of the local orthogonality type~\cite{dourdent25,fritz13}. Pairwise exclusive completeness is in fact a \textit{structural}~\cite{kunjwal2019b} reading of Specker's principle~\cite{specker,cabello12s}\textemdash ``if you have several questions and you can answer any two of them, then you can also answer all of them''\textemdash  in a Bell scenario, i.e.\ a structural version of the local orthogonality principle~\cite{fritz13,sainz14}. 
Known as \textit{orthocoherence}, it also plays a key role in the discussion of tensor products in quantum logic~\cite{foulisrandall81,coecke2000,bacciagaluppi2023}.
As we show in forthcoming work, pairwise exclusive completeness anchors in fact process functions within the hypergraph-theoretic language of contextuality scenarios~\cite{acin}. 

These verifiable membership criteria enable a more systematic study of logically consistent classical processes with indefinite causal order. They may, for instance, help establish the conjecture of
Ref.~\cite{tselentis23} that a ``siblings-on-cycles'' property suffices for a causal structure to be admissible. Moreover, the physical realizability of these consistent loops remains fundamentally tied to their spacetime embeddability~\cite{vilasini24}. Since process matrices\textemdash and by extension process functions\textemdash form a linear subset  of both post-selected closed timelike curves~\cite{politzer,bennett,svetlichny} and pre-and-post-selected quantum states~\cite{silva}, process functions can be simulated via quantum post-selection~\cite{araujo3}. Nevertheless, identifying the precise physical mechanisms that realize the full class of logically consistent classical loops\textemdash be it  non-trivial spacetime topologies or dynamical feedback architectures\textemdash remains a compelling direction for future research.

\medskip

\emph{Acknowledgments.} We thank Ognyan Oreshkov, Alastair Abbott, Cyril Branciard, Pierre Pocreau, Marco T\'ulio Quintino, Nasra Daher Ahmed, Alexei Grinbaum, Kuntal Sengupta, \"Amin Baumeler, and Zixuan Liu for enlightening discussions. We are also grateful to Fabio Costa, whose input helped us identify a loophole in an earlier version of our complete output‑reducibility proposal and develop the non‑erasing reducibility characterization (Theorem~\ref{th:nonerasingcharac}). We acknowledge the use of the AI assistant Claude (Anthropic), interaction with which helped identifying the loophole mentioned above, developing the non‑erasing reducibility characterization and proving Theorem~\ref{th:outputredphase}. The assistant was further used to improve the English of the manuscript. All mathematical statements and proofs were independently derived, checked, and are the sole responsibility of the authors.

H.D., A.L., and A.A. acknowledge financial support from the Government of Spain (Severo Ochoa CEX2019-000910-S, NextGenerationEU PRTR-C17.I1 and FUNQIP), Fundació Cellex, Fundació Mir-Puig, Generalitat de Catalunya (CERCA program). A.A. was also supported by the ERC AdG CERQUTE, the AXA Chair in Quantum Information Science. 
K.S., A.L., and E.C.B. acknowledge: This research was funded in whole or in part by the Austrian Science Fund (FWF) 10.55776/PAT4559623. For open access purposes, the author has applied a CC BY public copyright license to any author-accepted manuscript version arising from this submission. 
A.L. acknowledges the European Union (PASQuanS2.1, 101113690). 
R.K. acknowledges: This work received support from the French government under the France 2030 investment plan, as part of the Initiative d’Excellence d’Aix-Marseille Université-A*MIDEX, AMX-22-CEI-01.

\appendix

\section{Proof of Theorem \ref{th:nonerasingcharac}}\label{app:noerase}

For $n=1$, the statement follows directly from Theorem \ref{th:fixedpoint}: a unipartite process function is necessarily constant.
Let $n>1$. If $\bm w$ is a process function, then every reduction of
$\bm w$ by a local operation is again a process function (Theorem \ref{lem:characpf}). This holds,
in particular, for every non-erasing operation, proving the forward
implication.

For the converse we show the contrapositive: if the quasi-process $\bm w$ is not a process function, some non-erasing $f_k$ has a reduction $\bm w^{f_k}$ that is not well defined or not a valid $(n-1)-$partite process function, i.e.\ violates $(ii)$.

Assume first that the non-valid quasi-process $\w$ is self-signaling, i.e.\ there exists $k$ such that $x_k$ depends on $a_k$, i.e.\ $ \exists\ \a_{\backslash k}, a_k\neq a_k'$ such that
$x_k:=w_k(a_k,\bm a_{\backslash k})\;\neq\;x_k':=w_k(a_k',\bm a_{\backslash k}).$
Take a function $f_k$ such that $f_k(x'_k):=a_k$ and $f_k(x_k):=a_k'$ for $x_k\neq x'_k$. Then $f_k$ is non-erasing
(its image is $\{a_k,a_k'\}$, on which $w_k$ is non-constant), yet the one-party loop
$b=f_k(w_k(b,\bm a_{\backslash k}))$ has no solution: it takes values in $\{a_k,a_k'\}$,
and $a_k\mapsto a_k'$, $a_k'\mapsto a_k$. Hence $\bm w^{f_k}$ is not well defined and
$(ii)$ fails.

Assume now that the non-valid quasi-process $\w$ is non-self-signaling, i.e.\ every $\bm w^{f_k}$ is well-defined.
By Theorem~\ref{th:fixedpoint}, there exists a collection of local
operations $\bm f=(f_k)_{k\in\mathcal N}, \,\,  f_k:\mathcal X_k\to\mathcal A_k$, such that $\bm f\circ\bm w$ has no fixed point or several. Let us find a non-erasing operation $f_k$ such that $\bm f_{\backslash k}\circ\bm w^{f_k}$ has also no fixed point or several.

Let $\mathcal B_k:=\operatorname{Im}f_k$ (a proper subset of $\mathcal A_k$ in general, as $f_k$ need not be surjective) and $\mathcal B:=\prod_k\mathcal B_k$. Since $\bm f\circ\bm w$ takes values in $\mathcal B$, all its fixed points lie in $\mathcal B$, and $(\bm f\circ\bm w)|_{\mathcal B}$ has the same fixed points; in particular it still fails to have exactly one.

Let us assume that among the considered collection of $\bm f$, no $f_k$ is non-erasing for $\bm w$.  Intuitively, each $f_k$ then severs party $k$'s outgoing dependences, so that no signaling loop, and hence no paradox, can be detected by reducing with it. The precise definition follows from the negation of  Definition~\ref{def:nonerasing}: for every $j$, every $k\neq j$, every $\bm a_{\backslash k}$ and all $x_k,x_k'$, $w_j(f_k(x_k),\bm a_{\backslash k,j})=w_j(f_k(x_k'),\bm a_{\backslash k,j})$, where we have used $w_j(\bm a)=w_j(\bm a_{\backslash j})$ from non-self-signaling. Since $f_k(x_k)$ ranges over $\mathcal B_k$, this means that $a_k\mapsto w_j(a_k,\bm a_{\backslash k,j})$ is constant on $\mathcal B_k$ for every $\bm a_{\backslash k,j}\in\mathcal B_{\backslash k,j}$. Because $\mathcal B=\prod_k\mathcal B_k$, $\bm w$ is constant on $\mathcal B$. Therefore $(\bm f\circ\bm w)|_{\mathcal B}$ is constant and always has a unique fixed point, contradicting our assumption on $\bm f$. We conclude that some $f_k$ is non-erasing for $\bm w$.

Finally, let us show that the map $\bm a\mapsto\bm a_{\backslash k}$ is a bijection from the fixed points of $\bm f\circ\bm w$ onto those of $\bm f_{\backslash k}\circ\bm w^{f_k}$. 
A tuple $\bm a$ is a fixed point of $\bm f\circ\bm w$ if and only if
$a_k=f_k\bigl(w_k(\bm a_{\backslash k})\bigr)$ and $\bm a_{\backslash k}=\bm f_{\backslash k}\bigl(\bm w_{\backslash k}(a_k,\bm a_{\backslash k})\bigr).$
The first equation determines $a_k$ from $\bm a_{\backslash k}$ alone, so the map is injective; substituting it into the second gives exactly $\bm a_{\backslash k}=\bm f_{\backslash k}\bigl(\bm w^{f_k}(\bm a_{\backslash k})\bigr)$, so the image is contained in the fixed-point set of $\bm f_{\backslash k}\circ\bm w^{f_k}$. Conversely, given such a fixed point $\bm a_{\backslash k}$, setting $a_k:=f_k(w_k(\bm a_{\backslash k}))$ recovers a fixed point of $\bm f\circ\bm w$.

Hence, $\bm f_{\backslash k}\circ\bm w^{f_k}$ has the same number of fixed points as $\bm f\circ\bm w$, which is not unique by assumption. By Theorem~\ref{th:fixedpoint}, $\bm w^{f_k}$ is not a valid $(n-1)$-partite process function, and $(ii)$ fails.
\qed

\section{Process-matrix constraints and output-reducibility (Theorem \ref{th:outputrednogloop})}
\label{rem:averagefix}

The process-matrix constraint indexed by
$\mathcal K\subsetneq\mathcal N$ traces out the complementary systems
$A_j^{IO}$ and hence \emph{averages} over the outputs $a_j$,
$j\notin\mathcal K$. By contrast, output-reducibility (Proposition~\ref{prop:outputred}), requires validity for every \emph{fixed} assignment of these outputs. Indeed, for any $k_0\notin\mathcal K$, the sum over $a_{k_0}$ in Eq.~\eqref{eq:pmtolin} can be split into separate contributions, one for each value of $a_{k_0}$. Each contribution is precisely the output-reducibility constraint evaluated on the reduced function $\bm w^{a_{k_0}}$. 
Consequently, output-reducibility implies every process-matrix constraint indexed by a proper subset $\mathcal K\subsetneq\mathcal N$.

The converse does not hold constraint by constraint, because defects arising at different fixed outputs may cancel under averaging. For example, the $\mathcal K=\mathcal N$ process-matrix constraint Eq.~\eqref{eq:nogloballoop} may hold despite the presence of a global loop, in quasi-process also containing internal or local loops.
For example, $x_1=a_2\oplus a_3,\,\, x_2=a_1\oplus a_3,\,\, x_3=a_1\oplus a_2$ contains a global signaling loop, although Eq.~\eqref{eq:nogloballoop} is satisfied. However, it also contains internal loops: e.g.\ fixing $a_3=0$ yields
$x_1=a_2,\,\, x_2=a_1$. Output-reducibility detects the defect because it tests this reduced function separately.

Thus, the whole set of process-matrix constraints and output-reducibility completed with Eq.~\eqref{eq:nogloballoop} encode the same validity requirement when imposed as complete families, but their individual conditions do not correspond term by term. A process-matrix constraint indexed by $\mathcal K$ averages over the outputs of the complementary parties $\mathcal N\backslash\mathcal K$, whereas output-reducibility fixes those outputs and tests each resulting reduced function separately.

\bibliography{bibli}

@misc{af,
      note={M. Araújo and A. Feix, private communication (2014),
the process was communicated to Baumeler before it was
found by inspecting the extremal points of the non-causal
polytope characterized in Baumeler and Wolf \cite{baumeler2}},
}

@article{kunjwal23a,
  title = {Trading Causal Order for Locality},
  author = {Kunjwal, Ravi and Baumeler, \"Amin},
  journal = {Phys. Rev. Lett.},
  volume = {131},
  issue = {12},
  pages = {120201},
  numpages = {8},
  year = {2023},
  month = {Sep},
  publisher = {American Physical Society},
  doi = {10.1103/PhysRevLett.131.120201},
  url = {https://link.aps.org/doi/10.1103/PhysRevLett.131.120201}
}

@article{tselentis23,
   title="{Admissible Causal Structures and Correlations}",
   volume={4},
   ISSN={2691-3399},
   url={http://dx.doi.org/10.1103/PRXQuantum.4.040307},
   DOI={10.1103/prxquantum.4.040307},
   number={4},
   journal={PRX Quantum},
   publisher={American Physical Society (APS)},
   author={Tselentis, Eleftherios-Ermis and Baumeler, {\"A}min},
   year={2023},
   pages={040307},
   month=oct }

@misc{mejdoub25,
      title={Unilateral determination of causal order in a cyclic process}, 
      author={Ilyass Mejdoub and Augustin Vanrietvelde},
      year={2025},
      eprint={2506.18540},
      archivePrefix={arXiv},
      primaryClass={quant-ph},
      url={https://arxiv.org/abs/2506.18540}, 
}

@article{baumeler14,
  title = {Maximal incompatibility of locally classical behavior and global causal order in multiparty scenarios},
  author = {Baumeler, \"Amin and Feix, Adrien and Wolf, Stefan},
  journal = {Phys. Rev. A},
  volume = {90},
  issue = {4},
  pages = {042106},
  numpages = {7},
  year = {2014},
  month = {Oct},
  publisher = {American Physical Society},
  doi = {10.1103/PhysRevA.90.042106},
  url = {https://link.aps.org/doi/10.1103/PhysRevA.90.042106}
}

@article{araujo1,
   author = {{Ara{\'u}jo}, M. and {Branciard}, C. and {Costa}, F. and {Feix}, A. and 
	{Giarmatzi}, C. and {Brukner}, {\v C}.},
    title = "{Witnessing causal nonseparability}",
  journal = {New Journal of Physics},
     year = 2015,
    month = oct,
   volume = 17,
   number = 10,
      eid = {102001},
    pages = {102001},
      doi = {10.1088/1367-2630/17/10/102001},
 }

@ARTICLE{oreshkov1,
   author = {{Oreshkov}, O. and {Costa}, F. and {Brukner}, {\v C}.},
    title = "{Quantum correlations with no causal order}",
  journal = {Nat. Commun.},
     year = 2012,
    month = oct,
   volume = 3,
      eid = {1092},
    pages = {1092},
      doi = {10.1038/ncomms2076},
}

@article{oreshkov16,
	author = {Ognyan Oreshkov and Christina Giarmatzi},
	doi = {10.1088/1367-2630/18/9/093020},
	isbn = {1367-2630},
	journal = {New J. Phys.},
	number = {9},
	pages = {093020},
	title = {Causal and causally separable processes},
	ty = {JOUR},
	volume = {18},
	year = {2016}}

@article{novikov1989analysis,
  title={Analysis of the operation of a time machine},
  author={Novikov, Igor Dmitriyevich},
  journal={Zh. Eksp. Teor. Fiz.},
  volume={95},
  pages={769--776},
  year={1989}
}

@article{novikov,
  title = {Cauchy problem in spacetimes with closed timelike curves},
  author = {Friedman, John and Morris, Michael S. and Novikov, Igor D. and Echeverria, Fernando and Klinkhammer, Gunnar and Thorne, Kip S. and Yurtsever, Ulvi},
  journal = {Phys. Rev. D},
  volume = {42},
  issue = {6},
  pages = {1915--1930},
  numpages = {0},
  year = {1990},
  month = {Sep},
  publisher = {American Physical Society},
  doi = {10.1103/PhysRevD.42.1915},
  url = {https://link.aps.org/doi/10.1103/PhysRevD.42.1915},
}

@article{chiribella08,
	archiveprefix = {arXiv},
	author = {Chiribella, Giulio and D'Ariano, Giacomo Mauro and Perinotti, Paulo},
	doi = {10.1103/PhysRevLett.101.060401},
	eprint = {0712.1325},
	issue = {6},
	journal = {Phys. Rev. Lett.},
	month = {Aug},
	numpages = {4},
	pages = {060401},
	primaryclass = {quant-ph},
	publisher = {American Physical Society},
	title = {Quantum Circuit Architecture},
	url = {https://link.aps.org/doi/10.1103/PhysRevLett.101.060401},
	volume = {101},
	year = {2008}}

@article{chiribella09,
	adsurl = {http://adsabs.harvard.edu/abs/2009PhRvA..80b2339C},
	archiveprefix = {arXiv},
	author = {Giulio Chiribella and Giacomo Mauro D'Ariano and Paulo Perinotti},
	doi = {10.1103/PhysRevA.80.022339},
	eid = {022339},
	eprint = {0904.4483},
	journal = {Phys. Rev. A},
	month = aug,
	number = 2,
	pages = {022339},
	primaryclass = {quant-ph},
	title = {{Theoretical framework for quantum networks}},
	volume = 80,
	year = 2009}

@article{choi75,
	author = {Man-Duen Choi},
	doi = {10.1016/0024-3795(75)90075-0},
	journal = {Linear Algebra Appl.},
	pages = {285--290},
	title = {Completely positive linear maps on complex matrices},
	volume = {10},
	year = {1975}}

@article{fritz13,
	doi = {10.1038/ncomms3263},
	url = {https://doi.org/10.1038/ncomms3263},
	year = 2013,
	month = {aug},
	publisher = {Springer Science and Business Media {LLC}},
	volume = {4},
	number = {1},
	author = {T. Fritz and A.B. Sainz and R. Augusiak and J Bohr Brask and R. Chaves and A. Leverrier and A. Ac{\'{\i}}n},
	title = {Local orthogonality as a multipartite principle for quantum correlations},
	journal = {Nat. Commun.},
    pages = {2263}
}

@article{kunjwal2019b,
  doi = {10.22331/q-2019-09-09-184},
  url = {https://doi.org/10.22331/q-2019-09-09-184},
  title = {Beyond the {C}abello-{S}everini-{W}inter framework: {M}aking sense of contextuality without sharpness of measurements},
  author = {Kunjwal, Ravi},
  journal = {{Quantum}},
  issn = {2521-327X},
  publisher = {{Verein zur F{\"{o}}rderung des Open Access Publizierens in den Quantenwissenschaften}},
  volume = {3},
  pages = {184},
  month = sep,
  year = {2019}
}

@article{sainz14,
  title = {Exploring the local orthogonality principle},
  author = {Sainz, A. B. and Fritz, T. and Augusiak, R. and Brask, J. Bohr and Chaves, R. and Leverrier, A. and Ac\'{\i}n, A.},
  journal = {Phys. Rev. A},
  volume = {89},
  issue = {3},
  pages = {032117},
  numpages = {18},
  year = {2014},
  month = {Mar},
  publisher = {American Physical Society},
  doi = {10.1103/PhysRevA.89.032117},
  url = {https://link.aps.org/doi/10.1103/PhysRevA.89.032117}
}

@misc{cabello12s,
      title={Specker's fundamental principle of quantum mechanics}, 
      author={Adan Cabello},
      year={2012},
      eprint={1212.1756},
      archivePrefix={arXiv},
      primaryClass={quant-ph},
      url={https://arxiv.org/abs/1212.1756}, 
}

@article{baumeler16,
   title={Device-independent test of causal order and relations to fixed-points},
   volume={18},
   ISSN={1367-2630},
   url={http://dx.doi.org/10.1088/1367-2630/18/3/035014},
   DOI={10.1088/1367-2630/18/3/035014},
   number={3},
   journal={New J. Phys.},
   publisher={IOP Publishing},
   author={Baumeler, {\"A}min and Wolf, Stefan},
   year={2016},
   month=apr, pages={035014} }

@article{lanczos1924stationare,
  title={{\"U}ber eine station{\"a}re {K}osmologie im {S}inne der {E}insteinschen {G}ravitationstheorie},
  author={Lanczos, Kornel},
  journal={Z. Phys.},
  volume={21},
  number={1},
  pages={73--110},
  year={1924},
  publisher={Springer}
}

@book{godel49,
	title = "{A Remark About the Relationship Between Relativity Theory and Idealistic Philosophy}",
	year = {1949},
	publisher = {Harper \& Row},
	author = { Kurt G\"odel},
	editor = {Paul Arthur Schilpp},
}

@article{specker,
 author = {Ernst Specker},
 journal = {Dialectica},
 number = {2/3},
 pages = {239--246},
 publisher = {Wiley},
 title = {Die {L}ogik {N}icht {G}leichzeitig {E}ntscheidbarer {A}ussagen},
 volume = {14},
 year = {1960},
 doi          = {10.1111/j.1746-8361.1960.tb00422.x},
 note         = {Reprinted in E.~P. Specker, \emph{Selecta}
                  (Birkh\"{a}user Verlag, Basel, Switzerland, 1990), pp.~175--182.
                 English translation, by M.~P. Seevinck:
                  ``The logic of non-simultaneously decidable propositions'',
                  \href{https://arxiv.org/abs/1103.4537}{arXiv:1103.4537v3 [physics.hist-ph]}.}
}

@misc{liu25,
      title={Parity erasure: a foundational principle for indefinite causal order}, 
      author={Zixuan Liu and Ognyan Oreshkov},
      year={2026},
      eprint={2512.08635},
      archivePrefix={arXiv},
      primaryClass={quant-ph},
      url={https://arxiv.org/abs/2512.08635}, 
}

@misc{costa26,
      title="{Indefinite Quantum Causality}", 
      author={Fabio Costa and Giulia Rubino and Cyril Branciard and Časlav Brukner and Marco Túlio Quintino},
      year={2026},
      eprint={2606.19438},
      archivePrefix={arXiv},
      primaryClass={quant-ph},
      url={https://arxiv.org/abs/2606.19438}, 
}

@article{silva,
  author={{Silva}, R. and Yelena {Guryanova}, Y. and {Short}, A. and {Skrzypczyk}, P. and {Brunner}, N. and  {Popescu}, N.},
  title={Connecting processes with indefinite causal order and multi-time quantum states},
  journal={New J. Phys.},
  volume={19},
  number={10},
  pages={103022},
  url={http://stacks.iop.org/1367-2630/19/i=10/a=103022},
  year={2017},
}

@Article{acin,
author="Ac{\'i}n, A.
and Fritz, T.
and Leverrier, A.
and Sainz, A. B.",
title="A Combinatorial Approach to Nonlocality and Contextuality",
journal="Communications in Mathematical Physics",
year="2015",
month="Mar",
day="01",
volume="334",
number="2",
pages="533--628",
doi="10.1007/s00220-014-2260-1",
url="https://doi.org/10.1007/s00220-014-2260-1",
}

@book{foulisrandall81, title={Empirical logic and tensor products}, volume={5}, abstractNote={In our work we are developing a formalism called empirical logic to support a generalization of conventional statistics; the resulting generalization is called operational statistics. We are not attempting to develop or advocate any particular physical theory; rather we are formulating a precision 'language' in which such theories can be expressed, compared, evaluated, and related to laboratory experiments. We believe that only in such a language can the connections between real physical procedures (operations) and physical theories be made explicit and perspicuous. (orig./HSI)}, publisher={Bibliographisches Inst.}, author={Foulis, D.J. and Randall, C.H.}, year={1981}, pages={p. 9‑20.} }

@article{politzer,
  title = {Path integrals, density matrices, and information flow with closed timelike curves},
  author = {Politzer, H. David},
  journal = {Phys. Rev. D},
  volume = {49},
  issue = {8},
  pages = {3981--3989},
  numpages = {0},
  year = {1994},
  month = {Apr},
  publisher = {American Physical Society},
  doi = {10.1103/PhysRevD.49.3981},
  url = {https://link.aps.org/doi/10.1103/PhysRevD.49.3981}
}

@article{bennett,
  title = {Teleportation, Simulated Time Travel, and How to Flirt with Someone Who Has Fallen into a Black Hole},
  author = {Bennett, C. H. and Schumacher, B.},
  journal = {Talk in QUPON},
  year = {2005},
  url = {http://web.archive.org/web/20060207123032/http://www.research.ibm.com/people/b/bennetc/QUPONBshort.pdf}
}

@Article{svetlichny,
author="Svetlichny, George",
title="Time Travel: Deutsch vs. Teleportation",
journal="Int. J. Theor. Phys.",
year="2011",
month="Dec",
day="01",
volume="50",
number="12",
pages="3903--3914",
issn="1572-9575",
doi="10.1007/s10773-011-0973-x",
url="https://doi.org/10.1007/s10773-011-0973-x"
}

@article{araujo3,
  title = {Quantum computation with indefinite causal structures},
  author = {Ara\'ujo, Mateus and Gu\'erin, Philippe Allard and Baumeler, \"Amin},
  journal = {Phys. Rev. A},
  volume = {96},
  issue = {5},
  pages = {052315},
  numpages = {11},
  year = {2017},
  month = {Nov},
  publisher = {American Physical Society},
  doi = {10.1103/PhysRevA.96.052315},
  url = {https://link.aps.org/doi/10.1103/PhysRevA.96.052315}
}

@article{baumeler2,
	doi = {10.1088/1367-2630/18/1/013036},
	year = {2016},
	month = {jan},
	publisher = {{IOP} Publishing},
	volume = {18},
	number = {1},
	pages = {013036},
	author = {{\"A}min Baumeler and Stefan Wolf},
	title = {The space of logically consistent classical processes without causal order},
	journal = {New J. Phys.},
}

@misc{dourdent25,
      title={Paradox-free classical non-causality and unambiguous non-locality without entanglement are equivalent}, 
      author={Hippolyte Dourdent and Kyrylo Simonov and Andreas Leitherer and Emanuel-Cristian Boghiu and Ravi Kunjwal and Saronath Halder and Remigiusz Augusiak and Antonio Acín},
      year={2025},
      eprint={2512.23599},
      archivePrefix={arXiv},
      primaryClass={quant-ph},
      url={https://arxiv.org/abs/2512.23599}, 
}

@article{wechs,
	doi = {10.1088/1367-2630/aaf352},
	url = {https://doi.org/10.1088/1367-2630/aaf352},
	year = 2019,
	month = {jan},
	publisher = {{IOP} Publishing},
	volume = {21},
	number = {1},
	pages = {013027},
	author = {Julian Wechs and Alastair A Abbott and Cyril Branciard},
	title = {On the definition and characterisation of multipartite causal (non)separability},
	journal = {New Journal of Physics}
}

@article{cabello3,
  title = "{Graph-Theoretic Approach to Quantum Correlations}",
  author = {Cabello, Ad\'an and Severini, Simone and Winter, Andreas},
  journal = {Phys. Rev. Lett.},
  volume = {112},
  issue = {4},
  pages = {040401},
  numpages = {5},
  year = {2014},
  month = {Jan},
  publisher = {American Physical Society},
  doi = {10.1103/PhysRevLett.112.040401},
  url = {https://link.aps.org/doi/10.1103/PhysRevLett.112.040401}
}

@misc{chiribella2014,
      title={Measurement sharpness cuts nonlocality and contextuality in every physical theory}, 
      author={G. Chiribella and X. Yuan},
      year={2014},
      eprint={1404.3348},
      archivePrefix={arXiv},
      primaryClass={quant-ph},
      url={https://arxiv.org/abs/1404.3348}, 
}

@article{coecke2000,
  title={Operational quantum logic: An overview},
  author={Coecke, Bob and Moore, David and Wilce, Alexander},
  journal={Current Research in Operational Quantum Logic: Algebras, Categories, Languages},
  pages={1--36},
  year={2000},
  publisher={Springer},
  url={https://arxiv.org/abs/quant-ph/0008019}
}

@misc{bacciagaluppi2023,
      title="{A Proof of Specker's Principle}", 
      author={Guido Bacciagaluppi},
      year={2023},
      eprint={2305.07917},
      archivePrefix={arXiv},
      primaryClass={quant-ph},
      url={https://arxiv.org/abs/2305.07917}, 
}

@article{chiribella2016,
   title={Bridging the gap between general probabilistic theories and the device-independent framework for nonlocality and contextuality},
   volume={250},
   ISSN={0890-5401},
   url={http://dx.doi.org/10.1016/j.ic.2016.02.006},
   DOI={10.1016/j.ic.2016.02.006},
   journal={Information and Computation},
   publisher={Elsevier BV},
   author={Chiribella, Giulio and Yuan, Xiao},
   year={2016},
   month=Oct, pages={15–49} }

@article{hawking,
  title = {Chronology protection conjecture},
  author = {Hawking, S. W.},
  journal = {Phys. Rev. D},
  volume = {46},
  issue = {2},
  pages = {603--611},
  numpages = {0},
  year = {1992},
  month = {Jul},
  publisher = {American Physical Society},
  doi = {10.1103/PhysRevD.46.603},
  url = {https://link.aps.org/doi/10.1103/PhysRevD.46.603}
}

@article{vilasini24,
   title={Embedding cyclic information-theoretic structures in acyclic space-times: No-go results for indefinite causality},
   volume={110},
   ISSN={2469-9934},
   url={http://dx.doi.org/10.1103/PhysRevA.110.022227},
   DOI={10.1103/physreva.110.022227},
   number={2},
   journal={Phys. Rev. A},
   publisher={American Physical Society (APS)},
   author={Vilasini, V. and Renner, Renato},
   year={2024},
   pages={022227},
   month=aug }

@article{baumeler21,
   title={Equivalence of Grandfather and Information Antinomy Under Intervention},
   volume={340},
   ISSN={2075-2180},
   url={http://dx.doi.org/10.4204/EPTCS.340.1},
   DOI={10.4204/eptcs.340.1},
   journal={Electron. Proc. Theor. Comput. Sci.},
   publisher={Open Publishing Association},
   author={Baumeler, {\"A}min and Tselentis, Eleftherios-Ermis},
   year={2021},
   month={Sep},
   pages={1–12},
}

@misc{kunjwal23,
      title={Nonclassicality in correlations without causal order}, 
      author={Ravi Kunjwal and Ognyan Oreshkov},
      year={2025},
      eprint={2307.02565},
      archivePrefix={arXiv},
      primaryClass={quant-ph},
      url={https://arxiv.org/abs/2307.02565}, 
}

@article{baumeler19,
  title={Reversible time travel with freedom of choice},
  author={Baumeler, {\"A}min and Costa, Fabio and Ralph, Timothy C and Wolf, Stefan and Zych, Magdalena},
  journal={Class. Quantum Grav.},
  volume={36},
  number={22},
  pages={224002},
  year={2019},
  doi={10.1088/1361-6382/ab4973},
  publisher={IOP Publishing}
}

@article{baumeler22,
  title={Unlimited non-causal correlations and their relation to non-locality},
  author={Baumeler, {\"A}min and Gilani, Amin Shiraz and Rashid, Jibran},
  journal={Quantum},
  volume={6},
  pages={673},
  year={2022},
  doi={10.22331/q-2022-03-29-673},
  publisher={Verein zur F{\"o}rderung des Open Access Publizierens in den Quantenwissenschaften}
}

@article{tobar,
	doi = {10.1088/1361-6382/aba4bc},
	url = {https://doi.org/10.1088%2F1361-6382%2Faba4bc},
	year = 2020,
	month = {sep},
	publisher = {{IOP} Publishing},
	volume = {37},
	number = {20},
	pages = {205011},
	author = {Germain Tobar and Fabio Costa},
	title = {Reversible dynamics with closed time-like curves and freedom of choice},
	journal = {Class. Quantum Grav.},
}

@article{almeida2010guess,
  title={Guess your neighbor's input: A multipartite nonlocal game with no quantum advantage},
  author={Almeida, Mafalda L and Bancal, Jean-Daniel and Brunner, Nicolas and Ac{\'\i}n, Antonio and Gisin, Nicolas and Pironio, Stefano},
  journal={Phys. Rev. Lett.},
  volume={104},
  number={23},
  pages={230404},
  year={2010},
  doi={10.1103/PhysRevLett.104.230404},
  publisher={APS}
}

\clearpage
\onecolumngrid
\begin{center}\textbf{\large Supplemental Material}\end{center}

\section{Process functions as process matrices}\label{app:commutred}

Process functions are the deterministic classical counterpart of
process matrices~\cite{oreshkov1}, which are the Choi representations~\cite{choi75} of quantum supermaps that yield valid probabilities under arbitrary quantum operations.\medskip

Each party $k$ performs a single local quantum operation 
\begin{equation}
    M_{|i_k}: \mathcal{L}(\HS^{A^{I}_k}) \times \mathcal{I}_k \rightarrow \mathcal{L}(\HS^{A^{O}_k}) \times \mathcal{O}_k,\label{eq:instru}
\end{equation}
where the classical spaces $\X_k$ and $\A_k$ are replaced by spaces of linear operators $\mathcal{L}(\HS^{A^{I}_k})$ and $\mathcal{L}(\HS^{A^{O}_k})$. Concretely, $M_{|i_k}$ is the Choi operator~\cite{choi75} of a quantum instrument
$M_{|i_k} := (M_{o_k|i_k}^{A^{I}_kA^{O}_k})_{o_k|i_k}$, where each element 
$M_{o_k|i_k}^{A^{I}_kA^{O}_k} \in \mathcal{L}(\HS^{A^{I}_kA^{O}_k}):=\mathcal{L}(\HS^{A^{I}_k}\otimes\HS^{A^{O}_k})$ for any $i_k\in\I_k$ satisfies
$M_{o_k|i_k}^{A^{I}_kA^{O}_k} \ge 0$ and 
$\Tr_{A^{O}_k} \sum_{o_k} M_{o_k|i_k}^{A^{I}_kA^{O}_k} = \openone^{A^{I}_k}$. \\

With the notations  $\mathbf{A} = A_1\ldots A_n$, $A_k = A_k^{I} A_k^{O}$,  $\bm{o} = (o_1, \ldots, o_n)$ and $\bm{i} = (i_1, \ldots, i_n)$,
within the process matrix framework, the correlations established by the parties are given by the probabilities

\begin{equation}
    P(\bm{o}|\bm{i})=\Tr\left[
\big(\bigotimes_{k=1}^n M_{o_k|i_k}^{A_k}\big)^T\cdot
W^\mathbf{A}
\right]=\bigotimes_{k=1}^n M_{o_k|i_k}^{A_k}*W^\mathbf{A},\label{eq:probapm0}
\end{equation}
where $W\in\mathcal{L}(\HS^{\mathbf{A}})$ is the ``process matrix", and $*$ denotes the ``link product''~\cite{chiribella08,chiribella09}. For two operators $M\in\L(\mathcal{H}^{XY})$ and $N\in\L(\mathcal{H}^{YZ})$, the latter is defined as $M*N=\Tr_{Y} [(M^{T_Y}\otimes\id^{Z})(\id^{X}\otimes N)] $ where $T_Y$ denotes  partial transposition on  $Y$. This product is both commutative (up to reordering of tensor products) and associative (as long as each space involved in the product appears at most twice); a tensor product $M^X\otimes N^Y$ is a special case with trivial link.
\medskip

For Eq.~\eqref{eq:probapm0} to define valid probabilities under any arbitrary operations $M_{|i_k}^{A_k}$, the process matrix $W^\mathbf{A}$ must be positive semidefinite $W^\mathbf{A}\geq 0$ and belong to a nontrivial subspace $\mathcal{L}^{\mathcal{N}}$
of $\mathcal{L}(\HS^{\mathbf{A}})$~\cite{araujo1,oreshkov16,wechs}:
\begin{equation}
    W\in\mathcal{L}^{\mathcal{N}} \Leftrightarrow \forall \mathcal{K}\subseteq\mathcal{N}, \forall \mathcal{K}\neq \emptyset,\hspace{3mm} _{\prod_{k\in\mathcal{K}}[1-A_k^{O}]A_{\mathcal{N}\backslash\mathcal{K}}^{IO}}W=0,
\end{equation}
with the trace-and-replace notation
\begin{eqnarray}\label{eq:trnrep}
    _{[1-X]}W &=& W - _XW,\notag \\
    \label{eq:TeR2} _XW &=& \frac{\id^{X}}{d_X}\otimes \Tr_X W.
\end{eqnarray}

By \cite[Theorem~4]{baumeler22}, any process function $\w: \mathcal{A} \to \mathcal{X}$ admits a unique encoding as  a process matrix:
\begin{equation}\label{eq:classicalPM}
    W^{\mathbf{A}}_{\w}=\sum_{\a} \bigotimes_k \ketbra{a_k}{a_k}^{A_{k}^{O}}\otimes\ketbra{w_k(\a_{\backslash k})}{w_k(\a_{\backslash k})}^{A_{k}^{I}}.
\end{equation}
Setting $|\I_k|=|\mathcal{O}_k|=1$, arbitrary local operations $f_k:\mathcal{X}_k\rightarrow \mathcal{A}_k$ and $g_{j\neq k}:\mathcal{X}_j\rightarrow \mathcal{A}_j$ are encoded as quantum channels 
\begin{align}
M_{k}^{A_k}&=\sum_{x_k\in\X_k}\ketbra{x_k}{x_k}^{A_{k}^{I}}\otimes
\ketbra{f_k(x_{k})}{f_k(x_{k})}^{A_{k}^{O}},\\
M_{j}^{A_j}&=  \sum_{x_j\in\X_j}\ketbra{x_j}{x_j}^{A_{j}^{I}}\otimes
\ketbra{g_j(x_{j})}{g_j(x_{j})}^{A_{j}^{O}}.
\end{align}

\section{From the no global loop to phase erasure constraints (Theorem \ref{th:outputredphase})}
\label{app:noglobaloop}

Let us show that the process matrix constraint Eq.~\eqref{eq:nogloballoop} under the matrix encoding Eq.~\eqref{eq:classicalPM}
translates into the phase erasure conditions Eq.~\eqref{eq:phaseerasure}.

\begin{proof}
Denote 
  $W_{\w}
  = \sum_{\bm a\in\mathcal A}
  P_{\bm a}\otimes Q_{\w(\bm a)}$,
where $
  P_{\bm a}
  :=
  \bigotimes_{k\in\mathcal N}
  \ketbra{a_k}{a_k}^{A_k^O},$ $Q_{\bm x}
  :=
  \bigotimes_{k\in\mathcal N}
  \ketbra{x_k}{x_k}^{A_k^I}.$
For each party $k$, define the average of
$\bm a\mapsto Q_{\bm w(\bm a)}$ over $a_k$, and the centering $C_k$ by
\begin{equation}
    (\mathbb E_k Q_{\bm w})(\bm a)
  :=
  \frac{1}{d_k}
  \sum_{a_k'=0}^{d_k-1}
  Q_{\bm w(a_k',\bm a_{\backslash k})},\hspace{3mm}
  \mathsf C_k:=\mathrm{id}-\mathbb E_k.
\end{equation}
 Using the trace-and-replace map  (Eq.~\eqref{eq:trnrep}) on $A_k^O$ and
\begin{equation}
    \Tr\ketbra{a_k}{a_k}=1,
  \qquad
  \frac{\mathbbm 1^{A_k^O}}{d_k}
  =
  \frac{1}{d_k}
  \sum_{a_k=0}^{d_k-1}
  \ketbra{a_k}{a_k}^{A_k^O},
\end{equation}
we obtain
\begin{align}
  {}_{A_k^O}W_{\bm w}
  &=
  \frac{1}{d_k}
  \sum_{\bm a_{\backslash k}}
  \sum_{a_k,a_k'}
  P_{(a_k,\bm a_{\backslash k})}
  \otimes
  Q_{\bm w(a_k',\bm a_{\backslash k})}
  \notag\\
  &=
  \sum_{\bm a}
  P_{\bm a}\otimes
  (\mathbb E_kQ_{\bm w})(\bm a).
\end{align}
Consequently,
\begin{equation}
     _{[1-{A_k^O}]}W_{\bm w}
  =
  \sum_{\bm a}
  P_{\bm a}\otimes
  (\mathsf C_kQ_{\bm w})(\bm a). 
\end{equation}
The averaging operators commute, and hence so do the centering operators.
It follows that
\begin{equation}
     {}_{\prod_k[1-A_k^O]}W_{\bm w}
  =
  \sum_{\bm a}
  P_{\bm a}\otimes
  \left(\prod_{k\in\mathcal N}\mathsf C_k\right)
  Q_{\bm w}(\bm a). \label{eq:pmtolin}
\end{equation}
Since the operators $P_{\bm a}$ are linearly independent,
Eq.~\eqref{eq:nogloballoop} is equivalent to
\begin{equation}
\left(\prod_{k\in\mathcal N}\mathsf C_k\right)
  Q_{\bm w}(\bm a)=0
  \qquad\forall\a\in\A.
\end{equation}
Moreover, $Q_{\bm w(\bm a)}
  =\sum_{\bm x\in\mathcal X}
  \delta_{\bm w(\bm a),\bm x}\,Q_{\bm x}$,
and the operators $Q_{\bm x}$ are linearly independent. Thus the matrix constraint separates into the scalar conditions
\begin{equation}
  \left(\prod_{k\in\mathcal N}\mathsf C_k\right)
  F_{\bm x}=0
  \qquad\forall\,\bm x\in\mathcal X,
\label{eq:centeredconstraint}
\end{equation}
where $ F_{\bm x}(\bm a):=\delta_{\bm w(\bm a),\bm x}.$

It remains to express Eq.~\eqref{eq:centeredconstraint} in the complex Fourier
basis. Let us remark that the latter has no intrinsic physical significance here;
it is merely a convenient simultaneous eigenbasis for the averaging
operators $\mathbb E_k$. 
The essential fact is that
$\mathsf C_k=\mathrm{id}-\mathbb E_k$ projects onto the subspace of
functions that are non-constant in $a_k$. Thus
$\prod_k\mathsf C_k$ selects the component that depends jointly on all outputs.\medskip

Label each output alphabet as
$\mathcal A_k\simeq\mathbb Z_{d_k}$ and define the characters
\begin{equation}
 \chi_{\bm\nu}(\bm a)
  :=
  \prod_{k=1}^n
  e^{2\pi i\nu_k a_k/d_k},
  \hspace{3mm}
  \nu_k\in\{0,\ldots,d_k-1\}.   
\end{equation}
The Fourier coefficients of $F_{\bm x}$ are
\begin{equation}
    \widehat F_{\bm x}(\bm\nu)
  :=
  \frac{1}{|\mathcal A|}
  \sum_{\bm a\in\mathcal A}
  F_{\bm x}(\bm a)\,
  \chi_{\bm\nu}(\bm a).
\end{equation}
Orthogonality,
$\frac{1}{d_k}
  \sum_{a_k=0}^{d_k-1}
  e^{2\pi i\nu_k a_k/d_k}
  = \delta_{\nu_k,0}$,
shows that $\mathbb E_k$ retains precisely the Fourier modes with $\nu_k=0$. Consequently, $\mathsf C_k=\mathrm{id}-\mathbb E_k$
retains precisely those modes with $\nu_k\neq0$. The product
$\prod_k\mathsf C_k$ therefore retains exactly the Fourier modes for
which every frequency is nonzero.

Equation~\eqref{eq:centeredconstraint} is thus equivalent to
\begin{equation}
    \widehat F_{\bm x}(\bm\nu)=0
  \hspace{2mm}
  \forall\,\bm x\in\mathcal X,
   \hspace{2mm}
  \forall\,\bm\nu\text{ such that }\nu_k\neq0\ \forall k.  
\end{equation}
Substituting
$F_{\bm x}(\bm a)=\delta_{\bm w(\bm a),\bm x}$ and omitting the
irrelevant normalization factor $|\mathcal A|^{-1}$ gives
\begin{equation}
  \sum_{\bm a}
  \delta_{\bm w(\bm a),\bm x}
  \prod_{k=1}^n
  e^{2\pi i\nu_k a_k/d_k}
  =0,
\end{equation}
for every $\bm x$ and every
$\bm\nu\in\prod_k\{1,\ldots,d_k-1\}$.

This is precisely
Eq.~\eqref{eq:phaseerasure}.
\end{proof}

For binary outputs, the result immediately reduces to the parity condition, since the only nonzero frequency is $\nu_k=1$ and $e^{2\pi i a_k/2}=(-1)^{a_k}.$ Hence Eq.~\eqref{eq:phaseerasure} becomes
\begin{equation}
    \sum_{\bm a}
\delta_{\bm w(\bm a),\bm x}
(-1)^{\sum_k a_k}=0
\qquad\forall\,\bm x.
\end{equation}

\section{Proof of Theorem \ref{th:phaseerasure}}

We show that, for a classical process matrix $W_{\bm w}$, the complete
family of process-matrix constraints
\begin{equation}
  {}_{\prod_{k\in\mathcal K}[1-A_k^O]\,
  A_{\mathcal N\backslash\mathcal K}^{IO}}
  W_{\bm w}=0
  \quad
  \forall\,\emptyset\neq\mathcal K\subseteq\mathcal N
  \label{eq:allPMconstraints}
\end{equation}
is equivalent to the phase-erasure conditions
\begin{equation}
  \sum_{\bm a\in\mathcal A}
  \prod_{k\in\mathcal K}
  \delta_{w_k(\bm a),x_k}\,
  e^{2\pi i\nu_k a_k/d_k}
  =0
\end{equation}
for every nonempty $\mathcal K\subseteq\mathcal N$, every
$\bm x_{\mathcal K}\in\mathcal X_{\mathcal K}$, and every collection of
nonzero frequencies $\bm\nu_{\mathcal K}
  \in\prod_{k\in\mathcal K}\{1,\ldots,d_k{-}1\}.$

\begin{proof}
The argument is the same as for the full-set constraint
$\mathcal K=\mathcal N$. Fix a nonempty subset
$\mathcal K\subseteq\mathcal N$. Applying
${}_{A_{\mathcal N\backslash\mathcal K}^{IO}}(\cdot)$ traces and replaces the
systems of the complementary parties. Ignoring nonzero normalization
factors and identity operators, the resulting diagonal matrix has
coefficients
 $F_{\bm x_{\mathcal K}}(\bm a_{\mathcal K})
  :=
  \sum_{\bm a_{\mathcal N\backslash\mathcal K}}
  \prod_{k\in\mathcal K}
  \delta_{w_k(\bm a),x_k}.$
Thus, the process-matrix constraint indexed by $\mathcal K$ becomes
\begin{equation}
 \left(\prod_{k\in\mathcal K}\mathsf C_k\right)
  F_{\bm x_{\mathcal K}}=0
  \qquad
  \forall\,\bm x_{\mathcal K},   
\end{equation}
where $\mathsf C_k=\mathrm{id}-\mathbb E_k$ removes the part that is
constant in $a_k$.

As shown for $\mathcal K=\mathcal N$, the product
$\prod_{k\in\mathcal K}\mathsf C_k$ selects precisely the Fourier terms
whose frequencies are nonzero for every $k\in\mathcal K$. The above
condition is therefore equivalent to
\begin{equation}
  \sum_{\bm a_{\mathcal K}}
  F_{\bm x_{\mathcal K}}(\bm a_{\mathcal K})
  \exp\!\left(
    2\pi i\sum_{k\in\mathcal K}\frac{\nu_k a_k}{d_k}
  \right)
  =0
\end{equation}
for every $\bm x_{\mathcal K}$ and every
$\nu_k\in\{1,\ldots,d_k-1\}$, \(k\in\mathcal K\).

Substituting the definition of
$F_{\bm x_{\mathcal K}}$ and combining the sums over
$\bm a_{\mathcal K}$ and
$\bm a_{\mathcal N\backslash\mathcal K}$ gives
\begin{align}
  0&=\sum_{\bm a\in\mathcal A}
  \left(
    \prod_{k\in\mathcal K}
    \delta_{w_k(\bm a),x_k}
  \right)
  \exp\!\left(
    2\pi i\sum_{k\in\mathcal K}\frac{\nu_k a_k}{d_k}
  \right)
  \notag\\
&=\prod_{k\in\mathcal K}
  \delta_{w_k(\bm a),x_k}\,
  e^{2\pi i\nu_k a_k/d_k}.
\end{align}
Since this holds for every nonempty
$\mathcal K\subseteq\mathcal N$, the complete family of process-matrix
constraints is equivalent to a complete family of phase-erasure conditions.
\end{proof}

\end{document}